\documentclass[conference]{IEEEtran}

\usepackage[T1]{fontenc}
\usepackage{cite}
\usepackage{url}
\usepackage{hyperref}
\usepackage{graphicx}
\usepackage{amsmath,amssymb,amsthm}
 \usepackage{booktabs}
 \usepackage{algorithm}
\usepackage{algorithmic}
\usepackage{pgfplots}
 \usepackage{tikz}
 \pgfplotsset{compat=1.18}

\newtheorem{definition}{Definition}

\newtheorem{theorem}{Theorem}

\newtheorem{remark}{Remark}

\newcommand{\E}{\mathbb{E}}
\newcommand{\Prob}{\mathbb{P}}

\usetikzlibrary{arrows.meta, positioning, calc, backgrounds, shadows, fit}

\title{High-Fidelity Network Management for Federated AI-as-a-Service: Cross-Domain Orchestration}

\author{
    Mohaned Chraiti\IEEEauthorrefmark{1}, Ozgur Ercetin\IEEEauthorrefmark{1}, and   Merve Saimler\IEEEauthorrefmark{2}\\
      \IEEEauthorrefmark{1}Electronics Engineering Department, Sabancı University\\
       \IEEEauthorrefmark{2} Ericsson Research, Türkiye\\
     Emails: mohaned.chraiti@sabanciuniv.edu, oercetin@sabanciuniv.edu, and merve.saimler@ericsson.com.
}

\begin{document}

\maketitle

\begin{abstract}


To support the emergence of AI-as-a-Service (AIaaS), communication service providers (CSPs) are on the verge of a radical transformation—from pure connectivity providers to AIaaS a managed network service (control-and-orchestration plane that exposes AI models). In this model, the CSP is responsible not only for transport/communications, but also for intent-to-model resolution and joint network–compute orchestration, i.e.,  reliable and timely end-to-end delivery. The resulting end-to-end AIaaS service thus becomes governed by communications impairments (delay, loss) and inference impairments (latency, error). A central open problem is an operational AIaaS control-and-orchestration framework that enforces high fidelity, particularly under multi-domain federation. This paper introduces an assurance-oriented AIaaS management plane based on Tail-Risk Envelopes (TREs): signed, composable per-domain descriptors that combine deterministic guardrails with stochastic rate–latency–impairment models. Using stochastic network calculus, we derive bounds on end-to-end delay violation probabilities across tandem domains and obtain an optimization-ready risk-budget decomposition. We show that tenant-level reservations prevent bursty traffic from inflating tail latency under TRE contracts. An auditing layer then uses runtime telemetry to estimate extreme-percentile performance, quantify uncertainty, and attribute tail-risk to each domain for accountability. Packet-level Monte-Carlo simulations demonstrate improved p99.9 compliance under overload via admission control and robust tenant isolation under correlated burstiness. 

\end{abstract}


\section{Introduction}
\subsection{Motivations}
Artificial intelligence is advancing at an unprecedented rate and, over the past three years, has transitioned from a specialized tool employed primarily by experts to a widely accessible, everyday technology—a shift driven predominantly by large language model (LLM) platforms. This acceleration is catalyzing the emergence of AI-as-a-Service (AIaaS) and, with it, an ongoing central leap in the sixth generation networks (6G) services: the communication service provider (CSP) is on the edge of evolving from a connectivity-only operator into a AI control-and-orchestration provider that exposes AI capabilities as a managed network service \cite{Merve1,SaimlerSlides2024,merve2}. This evolution is consistent with the AI-native network paradigm, in which AI capabilities are intrinsically embedded in network functions and exposed as consumable services rather than treated as external overlays. In this extended role, the operator does not merely transport traffic between AI model providers and end users; it performs intent-based service resolution, mapping user intents to an appropriate model and execution option, and orchestrates joint network and compute resources to meet contracted service targets. As a result, the end user is no longer bound to a specific, pre-known model endpoint; instead, the user specifies the task requirements, and the network selects the model and provisions the necessary communications and computational resources. The delivered service quality is inherently end-to-end, coupling transport impairments (delay, loss) with inference impairments (latency, model-error/quality degradation). Furthermore, the AIaaS orchestration leverages network-native intelligence, data, and Quality-of-Service (QoS) control that are unavailable to over-the-top AI providers. Consequently, enforceability must be addressed at the Control-and-Orchestration layer as a single service obligation.

\subsection{Related Work}

Although AI model providers have rapidly matured the AIaaS supply stack, the {network} exposure, control, and intent-driven orchestration of AIaaS remains at an infant stage. Nonetheless, research on contiguous topics such AI-native networking and zero-touch management architectures provide useful principles and interfaces, but they do not yet deliver an operational framework for enforceable end-to-end guarantees and multi-domain coordination.

AI-native network visions emphasize ubiquitous intelligence, distributed data infrastructure, and zero-touch operations, and advocate exposing AI capabilities and network-native services through Application-Programming-Interface (APIs)
\cite{SaimlerSlides2024,etsiZSM002,3gpp23501,3gpp23222capif,CAMARAProject}. From an AI-native perspective, the APIs are enablers, with AIaaS fundamentally relying on Machine Learning Operations (MLOps) as-a-Service to operationalize lifecycle management, monitoring, and assurance within the network \cite{SaimlerSlides2024,3gpp23222capif,CAMARAProject}. However, this API-centric view leaves a critical operational gap: \emph{how a CSP integrates exposed AIaaS features into a management and orchestration framework that is enforceable in the field} \cite{SaimlerSlides2024,etsi7772}.  APIs specify how services are requested, but they do not specify how the network can make \emph{verifiable commitments} about end-to-end performance when both transport and inference contribute to experienced latency and reliability.


Two requirements dominate AIaaS at scale. First, AIaaS contracts are naturally \emph{tail contracts}:
for interactive inference and agentic pipelines, Service Level Objectives (SLOs) are expressed as extreme-percentile deadlines
(e.g., p99/p99.9 end-to-end latency), not average latency \cite{DeanBarroso2013Tail}.
Average-QoS engineering is structurally insufficient because rare events arise from burst synchronization,
shared bottlenecks, scheduler transients, radio fades, and compute contention; moreover, these tail events may occur
in \emph{any} segment of the pipeline (access/transport, execution, or inter-domain traversal) while the client observes a single end-to-end SLO \cite{DeanBarroso2013Tail}.
Second, AIaaS requires \emph{federation}. Coverage, locality, and developer reach push AIaaS beyond a single operator
domain, creating multi-principal objectives, confidentiality constraints (no disclosure of internal topology, queue states,
or scheduler rules), and an accountability requirement where disputes and penalties must be supported by auditable
evidence \cite{etsi7772,LuoComMag2023SliceAudit}. Without an explicit representation of tail risk that composes across domains, an AIaaS orchestrator has only two failure
modes: over-provisioning (economically infeasible) or tail violations (trust-destroying) \cite{DeanBarroso2013Tail}.

\subsection{Problem Statement}
This paper targets the missing architectural primitive: a \emph{composable tail-risk contract} that (i) abstracts each
domain as a signed service descriptor, (ii) composes across tandem domains to yield end-to-end p99/p99.9 guarantees,
(iii) enables risk-budget decomposition and optimization under confidentiality constraints, and (iv) supports audit and
settlement without requiring internal disclosure. The core contribution is therefore not a new API surface; it is an
\emph{assurance layer} for AIaaS management and orchestration. This assurance layer complements ongoing AIaaS API standardization efforts by addressing enforceability and accountability rather than service invocation semantics.

\subsection{Contributions}
The paper makes four technical contributions that together operationalize federated AIaaS as an enforceable managed service.

\begin{itemize}
\item We introduce Tail-Risk Envelopes (TREs): signed, per-domain contracts for tail-latency assurance that are composable across domains with minimal disclosure.
\item We derive tractable bounds on per-domain and end-to-end delay-violation probabilities in tandem domains, enabling p99/p99.9 feasibility checks for admission and reservation.
\item We formulate federated provisioning as a risk-budgeted optimization that selects paths, decomposes end-to-end tail budgets across domains, and reserves per-tenant capacity, with a proof of burst-resilient isolation.
\item We develop a telemetry-driven auditing layer that estimates extreme percentiles with confidence intervals, updates TRE uncertainty when tails shift, and attributes tail-risk across domains for accountability.
\end{itemize}

The developed approach is designed to integrate with standardized AIaaS exposure and orchestration frameworks that enable deployment within existing and emerging 3GPP- and industry-aligned architectures.

\section{Contract Interface for Federated AIaaS}\label{sec:sysmodel}

\subsection{Federated AIaaS execution pipeline and service contract}\label{subsec:pipeline}
An AIaaS request is realized as a {multi-stage execution pipeline} that spans both
communication and computation: radio access, transport/backhaul, edge or cloud execution, and (when needed)
inter-operator transit. The key operational characteristic is {federation}: at one or more stages, the execution
domain can be selected from multiple administrative entities (e.g., alternative edge providers or partner operators),
subject to locality, trust, and compliance constraints, consistent with global API aggregation visions
\cite{EricssonGlobalNetworkAPI2023,GSMAOpenGateway,CAMARAProject,3gpp23222capif}.
This federated execution model directly supports multi-CSP AIaaS exposure and aggregation scenarios.

\paragraph{Pipeline stages and architecture}
We represent the AIaaS path as $L$ ordered {stages} indexed by $\ell\in\{1,\dots,L\}$.
Each stage $\ell$ offers a {set of feasible domains} $\mathcal{D}_\ell(\Omega_{u,k})$ determined by policy and tenant constraints
(e.g., allowed operators, locality region, trust level), where $\Omega_{u,k}$ is defined below.
A concrete execution path is therefore a {sequence of selected domains}
\begin{equation}
\pi \triangleq (d_1,\dots,d_L), \qquad d_\ell \in \mathcal{D}_\ell(\Omega_{u,k}).
\end{equation}
A domain $d$ may represent a Radio-Access-Network (RAN) segment, a transport segment, an edge cluster, a core cloud, or an inter-operator transit segment.
This distinction (stages vs.\ selectable domains per stage) makes federation explicit and separates {where choice exists}
from {where the pipeline is structurally fixed}.

Fig.~\ref{fig:sys_overview} summarizes the operational control loop.
A CSP hosts an AIaaS {control-and-orchestration (C\&O)} function that (i) receives intent-level requests
through exposure APIs, (ii) resolves intent-to-model and execution placement, (iii) provisions joint {network and compute} resources along a
selected federated path, and (iv) audits delivered performance for accountability and settlement. In practice, intent intake and service invocation are realized through standardized network exposure mechanisms, while the assurance logic operates below the API layer.
This aligns with AI-native and zero-touch management directions \cite{SaimlerSlides2024,etsiZSM002,etsi7772}, and is compatible with
existing exposure/API frameworks (e.g., NEF/CAPIF/CAMARA) \cite{3gpp23501,3gpp23222capif,CAMARAProject}.
The paper’s focus is the missing assurance layer: {how to make end-to-end tail guarantees enforceable and composable across domains without
requiring disclosure of internal scheduler or queue state}.

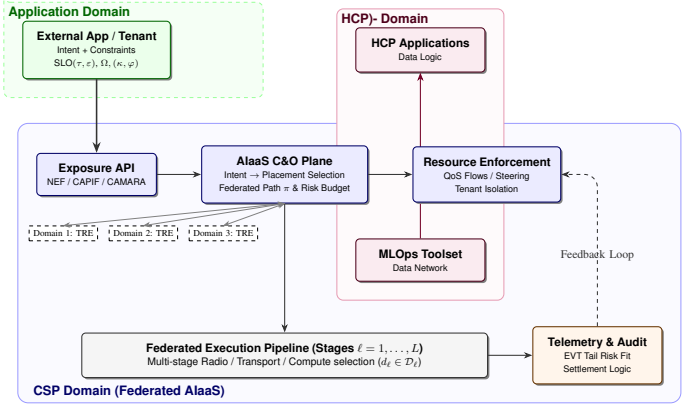
\begin{figure}
  \hspace{-0.3cm}
 \resizebox{0.5\textwidth}{!}{
\begin{tikzpicture}[
    font=\sffamily\small,
    >=Stealth,
    node distance=1.2cm and 1.2cm,
    base/.style={draw, thick, rounded corners=3pt, align=center, inner sep=8pt, fill=white, drop shadow={opacity=0.15, shadow xshift=0.5mm, shadow yshift=-0.5mm}},
    app_box/.style={base, fill=green!8, draw=green!40!black},
    csp_box/.style={base, fill=blue!8, draw=blue!40!black},
    exec_box/.style={base, fill=gray!8, draw=gray!40!black},
    audit_box/.style={base, fill=orange!8, draw=orange!40!black},
    hcp_box/.style={base, fill=purple!8, draw=purple!40!black},
    flow/.style={->, thick, color=black!80},
    step/.style={draw, circle, fill=red!70!black, text=white, font=\scriptsize\bfseries, inner sep=1.5pt, minimum size=13pt}
]

\node[app_box, minimum width=4cm] (app) {
    \textbf{External App / Tenant}\\
    \scriptsize Intent + Constraints\\
    \scriptsize SLO$(\tau,\varepsilon)$, $\Omega, (\kappa,\varphi)$
};

\node[csp_box, minimum width=3cm, below=2cm of app] (api) {
    \textbf{Exposure API}\\
    \scriptsize NEF / CAPIF / CAMARA
};

\node[csp_box, minimum width=4.5cm, right=of api] (co) {
    \textbf{AIaaS C\&O Plane}\\
    \scriptsize Intent $\rightarrow$ Placement Selection\\
    \scriptsize Federated Path $\pi$ \& Risk Budget
};

\node[draw, dashed, fill=white, font=\scriptsize, minimum width=1.8cm, below=0.6cm of co.south, xshift=-6cm] (tre1) {Domain 1: TRE};
\node[draw, dashed, fill=white, font=\scriptsize, minimum width=1.8cm, right=0.3cm of tre1] (tre2) {Domain 2: TRE};
\node[draw, dashed, fill=white, font=\scriptsize, minimum width=1.8cm, right=0.3cm of tre2] (tre3) {Domain 3: TRE};

\node[exec_box, minimum width=11cm, below=3.5cm of co] (pipe) {
    \textbf{Federated Execution Pipeline (Stages $\ell=1,\dots,L$)}\\
    \footnotesize Multi-stage Radio / Transport / Compute selection ($d_\ell \in \mathcal{D}_\ell$)
};

\node[audit_box, minimum width=3.5cm, right=of pipe] (audit) {
    \textbf{Telemetry \& Audit}\\
    \scriptsize EVT Tail Risk Fit\\
    \scriptsize Settlement Logic
};

\node[hcp_box, minimum width=3.5cm, right=5cm of app] (hcp_app) {
    \textbf{HCP Applications}\\
    \scriptsize Data Logic
};

\node[hcp_box, minimum width=3.5cm, below=4.5cm of hcp_app] (hcp_ml) {
    \textbf{MLOps Toolset}\\
    \scriptsize Data Network
};

\begin{scope}[on background layer]
    \filldraw[draw=green!40, fill=green!4, dashed, rounded corners=5pt] 
        ($(app.north west)+(-0.5,0.5)$)-| ($(app.south east)+(2.5,-0.5)$)-| cycle;
    \node[anchor=north west, font=\sffamily\bfseries\color{green!40!black}] at ($(app.north west)+(-0.5,0.5)$) {Application Domain};

    \filldraw[draw=blue!40, fill=blue!2, rounded corners=8pt] 
        ($(api.north west)+(-0.5,0.8)$) rectangle ($(audit.south east)+(0.5,-0.5)$);
    \node[anchor=north west, font=\sffamily\bfseries\color{blue!40!black}] at ($(api.north west)+(-0.2,-6.2)$) {CSP Domain (Federated AIaaS)};

    \filldraw[draw=purple!40, fill=purple!2, rounded corners=5pt] 
        ($(hcp_app.north west)+(-0.5,0.5)$) rectangle ($(hcp_ml.south east)+(0.5,-0.5)$);
    \node[anchor=north west, font=\sffamily\bfseries\color{purple!40!black}] at ($(hcp_app.north west)+(-0.5,0.5)$) {HCP)- Domain};
\end{scope}

\draw[flow, line width=1.2pt] (app.south) -- (api.north);

\draw[flow] (api.east) -- (co.west);

\draw[flow] (co.south) -- (co |- pipe.north);


\draw[flow] (pipe.east) -- (audit.west);


\draw[flow, color=purple!60!black] (hcp_ml.north) -- (hcp_app.south);

\foreach \i in {1,2,3} {
    \draw[<->, gray, thin] (tre\i.north) -- (co.270);
}

\node[csp_box, minimum width=3.5cm, right=of co] (enf) {
    \textbf{Resource Enforcement}\\
    \scriptsize QoS Flows / Steering\\
    \scriptsize Tenant Isolation
};

\draw[flow, dashed, rounded corners=10pt] (audit.north) |- node[above, font=\small, pos=0.2] {Feedback Loop} (enf.east);
\draw[flow] (co.east) -- (enf.west);
\end{tikzpicture}}
\vspace{-.4cm}
\caption{Federated AIaaS execution pipeline and assurance loop. The CSP AIaaS receives intents, selects model/placement, and provisions joint network--compute resources across a multi-stage path. Each administrative domain, including the Hyperscale Cloud Provider (HCP), exposes a TRE contract (signed, composable) rather than raw internal state. Telemetry feeds an audit layer for p99/p99.9 calibration.}
\label{fig:sys_overview}
\vspace{-0.5cm}
\end{figure}


\paragraph{Tenant-level AI service-level objective}
We formalize the tenant contract as an SLO centered on {tail latency} and correctness constraints.
For tenant $u$ and inference class $k$, define
\begin{equation}
\mathrm{SLO}_{u,k}=\big(\tau_{u,k},\varepsilon_{u,k},\Omega_{u,k},\kappa_{u,k},\varphi_{u,k}\big),
\label{eq:ai_slo}
\end{equation}
where: (i) $\tau_{u,k}$ is the end-to-end deadline; (ii) $\varepsilon_{u,k}$ is the maximum allowed violation probability
(e.g., $\varepsilon_{u,k}=10^{-3}$ corresponds to p99.9);
(iii) $\Omega_{u,k}$ encodes locality/trust/compliance constraints and the admissible operator/domain sets;
(iv) $\kappa_{u,k}$ encodes accuracy/quality constraints (e.g., minimum task success probability or maximum error);
and (v) $\varphi_{u,k}$ encodes freshness constraints (e.g., maximum staleness) managed via model life-cycle and data APIs
\cite{SaimlerSlides2024}.
The motivation for tail SLOs is practical: user-perceived reliability is dominated by rare slow events in large-scale services,
so average-QoS engineering is insufficient \cite{DeanBarroso2013Tail}.

Let $A_u(s,t)$ denote the cumulative AI workload arriving for tenant $u$ over $[s,t)$, measured in a normalized unit
(e.g., inference jobs or equivalent compute quanta). Along a selected path $\pi$, each stage $\ell$ provides a service process
$S_{d_\ell}(s,t)$ in the same unit. The resulting (virtual) sojourn-time delay $W_{u,k}$ is induced by the interaction of arrivals and the
(min-plus) concatenation of per-stage services. The management objective is to enforce the tail-latency constraint
\begin{equation}
\Prob\{W_{u,k}>\tau_{u,k}\}\le \varepsilon_{u,k},
\label{eq:chance_constraint1}
\end{equation}
{under} multi-tenant contention, domain uncertainty, and federation.
The expression in \eqref{eq:chance_constraint1} is the central operational requirement: it is the mathematical form of ``p99/p99.9 end-to-end guarantee''. 

\begin{remark} Violations of \eqref{eq:chance_constraint1} can originate anywhere in the chain (radio, transport, compute, inter-domain), yet the user observes a single end-to-end SLO. Moreover, in a federated setting domains will not reveal internal queue lengths,
scheduler rules, or topology. This motivates a {contract interface} that is both composable and confidentiality-preserving.
\end{remark}

\subsection{TREs as a composable per-domain contract}\label{subsec:tre}
We introduce the TRE as the operational interface between a domain and the AIaaS orchestrator.
A TRE is the domain’s published service contract at a given reservation level (e.g., slice class, priority class, CPU share):
it is (i) \emph{signed} for accountability, (ii) \emph{composable} across tandem stages for end-to-end reasoning, and (iii) \emph{minimal}
so it can be disclosed without exposing internal state. TREs are not exposed to application developers. They serve as internal, contract-level abstractions that enable enforceable AIaaS orchestration across administrative boundaries while remaining compatible with standardized AIaaS exposure frameworks.

\begin{definition}[Tail-Risk Envelope]
For domain $d$ and tilting parameter $\theta>0$, the TRE is
\begin{equation}
\mathrm{TRE}_d(\theta)=\big(R_d,T_d,\kappa_d,\eta_d\big),
\label{eq:tre_tuple}
\end{equation}
interpreted as a stochastic rate--latency service:
\begin{equation}
S_d(s,t)\;\ge\; R_d\big[(t-s)-T_d\big]^+ - I_d(s,t), \qquad \forall\, t\ge s,
\label{eq:impair_service}
\end{equation}
where $R_d$ and $T_d$ are deterministic guardrail parameters and the impairment process $I_d$ satisfies the Moment Generating Function (MGF) bound
\begin{equation}
\E\!\left[e^{\theta I_d(s,t)}\right]\le \exp\!\Big(\theta\kappa_d (t-s)+\theta\eta_d\Big),\qquad \forall\, t\ge s.
\label{eq:impair_mgf}
\end{equation}
\end{definition}


\noindent When $I_d\equiv 0$, \eqref{eq:impair_service} reduces to the classical rate--latency service curve.
The impairment term captures {residual uncertainty} beyond the deterministic guardrail---e.g., bursty cross-traffic,
scheduler jitter, radio/transport randomness, compute contention, and model-execution variability---in a form that yields
exponential tail bounds when combined with MGF-bounded arrivals.
Operationally, a domain publishes a {family} of TREs indexed by reservation level (priority/slice/CPU share),
and the orchestrator selects (domain, reservation) pairs so that the composed end-to-end service satisfies
\eqref{eq:chance_constraint1}.
Crucially, the orchestrator does not require raw internal state (queue occupancy, per-flow scheduling details);
it only requires the TRE parameters, which are contract-level quantities.

Federation requires that domains can coordinate without disclosing internals.
A TRE provides exactly what is needed for orchestration:
(i) a deterministic baseline ($R_d,T_d$) that can be enforced by admission and reservation,
(ii) a stochastic uncertainty term ($\kappa_d,\eta_d$) that quantifies tail risk,
and (iii) a parameterization that is composable across stages and amenable to optimization.
Later sections show how TRE composition yields end-to-end violation bounds and how tail-risk budgets can be decomposed across domains
for federated provisioning.

Considering a standard MGF-style arrival envelope yields optimization-ready tail bounds while supporting admission and isolation constraints:  tenant-\(u\) arrivals satisfy, for some \(\theta>0\) and parameters \((\rho_u,\sigma_u)\),
\begin{equation}
\E\!\left[e^{\theta A_u(s,t)}\right]\le \exp\!\Big(\theta\rho_u(t-s)+\theta\sigma_u\Big),\qquad \forall\, t\ge s,
\label{eq:arr_mgf}
\end{equation}
which is consistent with effective-bandwidth formulations \cite{Chang2000}. Here, \(\rho_u\) captures sustained load and \(\sigma_u\) captures burstiness. One can hence derive an exponential tail bounds via stochastic network calculus, yet structured enough for (i) admission control, (ii) per-tenant reservation with isolation guarantees, and (iii) federated provisioning where domains exchange compact contract parameters rather than raw traces. In subsequent analysis, \eqref{eq:arr_mgf} is combined with TRE relations \eqref{eq:tre_tuple}--\eqref{eq:impair_mgf} to derive end-to-end delay-violation bounds for \eqref{eq:chance_constraint1} and construct cross-domain risk-budget decompositions.

\section{Compositional Tail Guarantees}\label{sec:snc}

This section aims to transform the published TRE parameters into {enforceable} p99/p99.9 guarantees. The key principle is that a domain does not expose internal queue states or scheduler rules; instead, it exposes a TRE contract \(\mathrm{TRE}_d(\theta)=(R_d,T_d,\kappa_d,\eta_d)\). Given tenant arrivals satisfying the \((\sigma,\rho)\)-MGF constraint, stochastic network calculus yields an explicit exponential bound on delay-violation probability. This bound is (i) {compositional} across tandem stages, (ii) {optimization-ready} as a chance constraint, and (iii) {federation-friendly} because it depends only on contract parameters.


\subsection{Single-domain tail bound under stochastic rate--latency}\label{subsec:single}
Consider tenant \(u\) traversing domain \(d\).
Recall the impaired service model in \eqref{eq:impair_service}--\eqref{eq:impair_mgf}
and the arrival constraint in \eqref{eq:arr_mgf}.
Define the \emph{net margin}
\begin{equation}
\Delta_{u,d}\triangleq R_d-\rho_u-\kappa_d .
\label{eq:margin}
\end{equation}
The condition \(\Delta_{u,d}>0\) is the feasibility threshold: if the sustained load \(\rho_u\) plus the impairment slope \(\kappa_d\)
meets or exceeds the reserved rate \(R_d\), exponential tail decay cannot be guaranteed.
From \eqref{eq:impair_service}, we have
\[
-S_d(s,t)\le -R_d[(t-s)-T_d]^+ + I_d(s,t).
\]
Applying the exponential operator on both sides and taking expectation,
\begin{equation}
\begin{aligned}
&\E\!\left[e^{-\theta S_d(s,t)}\right]\\
&\le \exp\!\big(-\theta R_d[(t-s)-T_d]^+\big)\;\E\!\left[e^{\theta I_d(s,t)}\right] \\
&\le \exp\!\Big(-\theta R_d[(t-s)-T_d]^+ + \theta\kappa_d(t-s)+\theta\eta_d\Big),
\end{aligned}
\label{eq:neg_service_mgf}
\end{equation}
where the last inequality stems from \eqref{eq:impair_mgf} for $\forall\,t\ge s$.
We assume that each domain provides FIFO-equivalent service at the contract level, which is required for backlog–delay equivalence and is standard in stochastic network calculus. Time is treated in discrete or discretized form for bounding purposes, consistent with the union-bound arguments employed below.
\begin{theorem}[Single-domain delay-violation bound]\label{thm:single}
Under \eqref{eq:arr_mgf} and \eqref{eq:neg_service_mgf}, if \(\Delta_{u,d}>0\), then for any deadline \(\tau\ge T_d\),

\begin{equation}
\begin{aligned}
&\Prob\{W_{u}^{(d)}>\tau\}\\
&\le\;
\underbrace{\frac{\exp\!\big(\theta(\sigma_u+\eta_d+\kappa_d T_d)\big)}
{1-\exp\!\big(-\theta\Delta_{u,d}\big)}}_{\text{Prefactor: burstiness \& uncertainty}}
\;\times\;
\underbrace{\exp\!\big(-\theta\Delta_{u,d}(\tau-T_d)\big)}_{\text{Tail-decay term}} .
\end{aligned}
\label{eq:single_delay_bound}
\end{equation}

\end{theorem}

\begin{proof}
     Let $B(t)$ denote backlog at time $t$ in a First-IN First-Out (FIFO) system, and $W(t)$ the virtual delay. Standard Stochastic Network Calculus (SNC) arguments relate delay to backlog: $W(t)>\tau$ implies that backlog exceeds the service that can be provided over $\tau$ slots. Formally, in a stable FIFO system, we have
\begin{equation}
\{W(t)>\tau\}\subseteq \Big\{\sup_{0\le s\le t}\big(A(s,t)-S(s,t)\big) > R(\tau-T)^+\Big\}.
\end{equation}
Applying Chernoff’s bound gives:
\begin{align}
\Prob\{A(s,t)-S(s,t)>x\}
&\le e^{-\theta x}\E[e^{\theta A(s,t)}]\E[e^{-\theta S(s,t)}].
\end{align}
Using \eqref{eq:arr_mgf} and \eqref{eq:neg_service_mgf}, for $\tau\ge T_d$ and $t-s=\ell$, we have
\begin{equation}
\begin{cases}
\E[e^{\theta A(s,t)}]\le e^{\theta\rho_u\ell+\theta\sigma_u}\\
\E[e^{-\theta S_d(s,t)}]\le e^{-\theta R_d(\ell-T_d)+\theta\kappa_d\ell+\theta\eta_d}
\end{cases},
\end{equation}
which gives
\begin{equation}
\begin{aligned}
&\Prob\{A(s,t)-S(s,t)>x\}\\
&\le \exp\!\Big(-\theta x+\theta\sigma_u+\theta\eta_d+\theta(\rho_u+\kappa_d-R_d)\ell+\theta R_d T_d\Big).
\end{aligned}
\end{equation}
A union bound is then appied to all feasible durations $\ell=t-s\in\{0,1,2,\dots\}$ and set $x=R_d(\tau-T_d)$ to obtain a geometric series with ratio $e^{-\theta\Delta_{u,d}}$, where $\Delta_{u,d}=R_d-\rho_u-\kappa_d$. Summing the series yields \eqref{eq:single_delay_bound}. 
\end{proof}


The bound in \eqref{eq:single_delay_bound} naturally separates into two factors:
a \emph{prefactor} and an \emph{exponential tail-decay term}.
The prefactor captures the ``baseline difficulty'' of meeting the SLO: it increases with the tenant burstiness
\(\sigma_u\) and with the domain’s uncertainty offset \(\eta_d\), both of which make rare p99/p99.9 excursions more likely.
The exponential term captures how fast the violation probability decays when the system has slack:
its exponent \(\theta \Delta_{u,d}(\tau-T_d)\) grows with the available service margin \(\Delta_{u,d}\) and with the deadline slack
\((\tau-T_d)\). Hence, reserving more service rate \(R_d\) increases \(\Delta_{u,d}\) and yields an exponential reduction of tail violations,
making the bound directly actionable for admission and reservation.


\subsection{End-to-end composition across an AIaaS pipeline: from per-domain TREs to a path-level bound}\label{subsec:tandem}

An AIaaS request traverses a sequence of stages (radio/transport/compute), so the end-to-end delay is determined by the
{slowest cumulative progress} along the chain. This reflects the AIaaS model, where both communication and AI execution stages jointly shape user-perceived performance. Intuitively, stage $\ell+1$ cannot start serving a job before stage $\ell$
has produced it; therefore, end-to-end service is governed by the {most constraining handoff} between stages.
Network calculus formalizes this tandem coupling via the \emph{min-plus} composition rule \cite{LeBoudecThiran2001NetCalc}.
For deterministic rate--latency stages, this rule yields two simple engineering consequences:
(i) latencies accumulate across stages, and (ii) the effective service rate is limited by the bottleneck stage.

In the considered setting, each domain $d$ exposes a TRE at tilting parameter $\theta$, $\mathrm{TRE}_d(\theta)=(R_d,T_d,\kappa_d,\eta_d)$, which lower-bounds its offered service by a rate--latency guardrail
minus an impairment term. When multiple domains are concatenated, a conservative (contract-safe) composition treats the impairments as additive across stages; this preserves enforceability without any disclosure of internal queue states or scheduler rules.

\noindent Let the selected path be $\pi=(d_1,\dots,d_L)$. We aggregate the published contract parameters into path-level descriptors:
\begin{equation}
\begin{aligned}
&R_{\min}\triangleq \min_{\ell\in\{1,\dots,L\}} R_{d_\ell},\qquad
T_{\Sigma}\triangleq \sum_{\ell=1}^{L} T_{d_\ell},\\
&
\kappa_{\Sigma}\triangleq \sum_{\ell=1}^{L}\kappa_{d_\ell},\qquad
\eta_{\Sigma}\triangleq \sum_{\ell=1}^{L}\left(\eta_{d_\ell}+\kappa_{d_\ell} T_{d_\ell}\right).
\label{eq:agg_params}
\end{aligned}
\end{equation}
The interpretation is direct: $T_\Sigma$ is the accumulated ``pipeline latency floor'';
$R_{\min}$ is the bottleneck reserved rate; and $(\kappa_\Sigma,\eta_\Sigma)$ quantify the accumulated tail uncertainty
induced by network/compute impairments across domains. 

Combining these aggregates with the tenant arrival descriptor $(\rho_u,\sigma_u)$ yields the end-to-end feasibility margin
\begin{equation}
\Delta_{u} \triangleq R_{\min}-\rho_u-\kappa_{\Sigma}.
\label{eq:agg_margin}
\end{equation}
A strictly positive margin $\Delta_u>0$ is necessary: if the sustained load $\rho_u$ plus the aggregate impairment slope
$\kappa_\Sigma$ meets or exceeds the bottleneck reserved rate $R_{\min}$, then no exponential tail guarantee can be maintained. With $\Delta_u>0$, substituting the path-level descriptors into the single-domain Chernoff/SNC steps
yields an explicit bound on end-to-end delay violation. For any deadline $\tau\ge T_{\Sigma}$,
\begin{equation}
\begin{aligned}
&\Prob\{W_u>\tau\}\le
\underbrace{\frac{\exp\!\big(\theta(\sigma_u+\eta_{\Sigma})\big)}
     {1-\exp\!\big(-\theta\Delta_{u}\big)}}_{\text{burst \& uncertainty inflation}}
\;\underbrace{\exp\!\big(-\theta\Delta_{u}(\tau-T_{\Sigma})\big)}_{\text{end-to-end tail decay}}.
\end{aligned}
\label{eq:tandem_bound}
\end{equation}

\noindent The bound in~\eqref{eq:tandem_bound} shows that a p99/p99.9-type end-to-end guarantee can be verified using only
(i) the tenant traffic descriptor $(\rho_u,\sigma_u)$ and (ii) the per-domain published TRE parameters along $\pi$. This property is critical for cross-operator federation, where internal information cannot be shared. No internal topology, queue lengths, or scheduler rules are required, which is the abstraction needed for
cross-operator federation.

Given a tail SLO \((\tau_u,\varepsilon_u)\) (e.g., \(\varepsilon_u=10^{-3}\) for p99.9),
a sufficient feasibility condition from \eqref{eq:tandem_bound} is
\begin{equation}
\begin{aligned}
&\theta\Delta_u(\tau_u-T_{\Sigma})\ge \log\!\frac{1}{\varepsilon_u}
+\theta(\sigma_u+\eta_{\Sigma})
+\log\!\frac{1}{1-e^{-\theta\Delta_u}}.
\end{aligned}
\label{eq:chance_constraint_ready}
\end{equation}
The right-hand side is a \emph{risk requirement}: it aggregates the demanded reliability level
$\log\!\big(1/\epsilon_u\big)$, the burst/uncertainty inflation term
$\theta\big(\sigma_u+\eta_\Sigma\big)$, and the finite-margin correction
$\log\!\Big(1/\big(1-e^{-\theta\Delta_u}\big)\Big)$.
The left-hand side is a \emph{risk capacity}: it scales with the effective slack $\Delta_u$
and the deadline slack $(\tau_u-T_\Sigma)$.
This form is directly actionable for orchestration: increasing the bottleneck reservation
$R_{\min}$ increases $\Delta_u$; reducing accumulated latency $T_\Sigma$ increases
$(\tau_u-T_\Sigma)$; and telemetry-driven auditing primarily updates $\eta_\Sigma$, tightening feasibility until additional reservation or rerouting is applied.






\section{Federated Provisioning and Assurance}\label{sec:fed_assurance}
Section~\ref{sec:snc} yields explicit p99/p99.9 tail-violation bounds from published TREs and tenant traffic descriptors. 
This section builds the corresponding {operational} management plane for multi-domain and multi-operator AIaaS: the orchestrator selects a
federated execution path, allocates a tail-risk budget across domains, reserves joint network--compute resources with strict multi-tenant
isolation, and coordinates operators without disclosing internal queue states or scheduler rules. Assurance is then closed by an auditing and
settlement loop that (i) estimates extreme percentiles from telemetry using Extreme Value Theory (EVT), (ii) conservatively updates TRE
uncertainty terms when the tail regime shifts, and (iii) attributes penalties and revenue to domains according to their marginal contribution
to end-to-end tail risk. The central principle is end-to-end enforceability: the client observes a single SLO, while the federation operates
through signed, composable contracts rather than internal-state disclosure.
\subsection{Federated provisioning under TRE contracts}
A tenant contract specifies \((\tau_u,\varepsilon_u)\) together with policy constraints \(\Omega\) (allowed operators/locations/trust),
and the orchestrator selects a path \(\pi=(d_1,\dots,d_L)\) through radio/transport/compute stages. Since the guarantee is end-to-end, a
practical management primitive is to decompose the violation budget across the domains of the selected path $
\sum_{\ell=1}^{L}\varepsilon_{u,d{_\ell}}\le \varepsilon_u$.  Each domain then enforces a local tail constraint of the form \eqref{eq:single_delay_bound} with its assigned $varepsilon_{u,d_\ell}$.
This decomposition provides (i) per-domain admission and accountability and (ii) a cost-aware lever to shape tail risk by tightening budgets
where reliability is scarce or expensive.

Tail guarantees are fragile under bursty cross-traffic unless isolation is explicit. We therefore enforce per-tenant effective-bandwidth
partitioning within each domain. Let \(\overline{R}_{d,u}\) denote the reserved service share assigned to tenant \(u\) in domain \(d\)
(e.g., a slice class, QoS-flow reservation, or computational resource share mapped to an equivalent service rate). Isolation is enforced by
\begin{equation}
\rho_u \le \overline{R}_{d,u}-\kappa_d,
\qquad
\sum_{u}\overline{R}_{d,u}\le R_d ,
\label{eq:eb_partition}
\end{equation}
which preserves a strictly positive margin \(\overline{R}_{d,u}-\rho_u-\kappa_d\) for each tenant independent of other tenants' burst
parameters. 

Federated provisioning is posed as a joint selection and reservation problem. Each domain \(d\) publishes (i) TRE parameters
\((T_d,\kappa_d,\eta_d)\) for each reservation class, (ii) a cost curve \(c_d(\overline{R}_{d,u})\) for rate/compute reservations (or a
mapping from CPU/GPU shares to an equivalent \(R_d\)), and (iii) admissibility constraints encoded by \(\Omega\). The orchestrator selects a path \(\pi\), reservations \(\{\overline{R}_{d,u}\}\), and risk budgets \(\{\varepsilon_{u,d}\}\) to minimize total cost while satisfying
tail constraints, budgets, and isolation:
\begin{equation}\label{eq:fedopt}
\begin{aligned}
\underset{\pi,\{\overline{R}_{d,u}\},\{\varepsilon_{u,d}\}}{\text{minimize}}& \sum_{d\in\pi}\sum_{u} c_d(\overline{R}_{d,u}) \\
\text{subject to}\quad
& \Prob\{W_u^{(d)}>\tau_u\} \le \varepsilon_{u,d}, \,\, \sum_{d\in\pi}\varepsilon_{u,d} \le \varepsilon_u,
\\ &\rho_u \le \overline{R}_{d,u}-\kappa_d,\, \sum_{u}\overline{R}_{d,u} \le R_d.
\end{aligned}
\end{equation}
\noindent The decisive point is that the feasible region is {tail-risk geometry} induced by TRE contracts and arrival descriptors, not mean-rate
QoS. Consequently, feasibility and reservation sizing can be evaluated from published contracts without internal queue states or scheduler
disclosure.

Federation further imposes confidentiality: a broker cannot request internal operator state. Therefore,  coordinate provisioning can be done via Alternating Direction Method of Multipliers (ADMM)
\cite{BoydADMM2011}. Each domain solves a private subproblem using only its own cost and capacity constraints; the orchestrator updates risk
budgets and consensus variables. For online operation, the ADMM loop runs on a slower control timescale (seconds--minutes) while domains
enforce the resulting reservations on fast schedulers; coupling the slow loop with drift-plus-penalty mechanisms yields stability and
near-optimality guarantees under stochastic arrivals \cite{Neely2010}. The novelty here is that the coupling constraints are TRE-induced
p99/p99.9 chance constraints rather than average-delay objectives. The optimization in \eqref{eq:fedopt} is implemented by Algorithm~\ref{alg:admm_evt}, which performs distributed broker--domain updates with privacy-preserving local solves.

\begin{algorithm}[h!]
\caption{Federated provisioning with ADMM and EVT assurance}
\label{alg:admm_evt}
\begin{algorithmic}[1]
\STATE \textbf{Input:} feasible paths \(\Omega\), TREs, domain costs \(c_d(\cdot)\), SLOs \((\tau_u,\varepsilon_u)\)
\STATE Initialize allocations \(\{\bar R_{d,u}\}\), risk budgets \(\{\varepsilon_{u,d}\}\), and dual variables
\REPEAT
  \STATE \textbf{Domain step (parallel):} each \(d\) solves its local ADMM subproblem for \(\{\bar R_{d,u}\}\)
  \STATE \textbf{Broker step:} update \(\{\varepsilon_{u,d}\}\) and consensus variables to satisfy \eqref{eq:fedopt}
  \STATE Update dual variables and residuals
\UNTIL primal/dual residuals are below tolerance or deadline is reached
\STATE Deploy reservations and start telemetry collection
\STATE Fit EVT tails from exceedances; if SLO violation risk increases, re-enter at line 3
\end{algorithmic}
\end{algorithm}


\subsection{Audit closure and settlement}
SNC bounds provide conservative contractual guarantees; operations additionally require auditing extreme percentiles from telemetry and updating
uncertainty terms when regimes change (overload, failures, execution drift). We use EVT in a peak-over-threshold (POT) form. Such auditing is essential to maintain trust in AIaaS guarantees over time. Given observed
end-to-end sojourn times \(\{L_i\}\), choose a high threshold \(q\) (e.g., empirical p98), form exceedances \(Y_i=L_i-q\) conditioned on
\(L_i>q\), and fit a generalized Pareto distribution (GPD) with shape \(\xi\) and scale \(\beta\)
\cite{Coles2001EVT,Pickands1975POT,Embrechts1997Extremes}. The POT estimator yields extreme quantiles \(Q_p\) for \(p>F(q)\), including p99
and p99.9, together with confidence intervals.

The assurance plane uses EVT for two purposes. First, it performs compliance verification: audited \(Q_{0.999}\) is compared against the
contract-implied risk from \eqref{eq:tandem_bound} (or its domain-level decompositions) to detect tail regressions and to distinguish
transient violations from systematic drift. Second, it enables conservative contract updates: when EVT indicates heavier tails or increased
variance, uncertainty parameters (primarily \(\eta_d\), and when necessary \(\kappa_d\)) are increased so that \eqref{eq:tandem_bound}
remains an upper bound at the target confidence level. Domains then re-issue updated, signed TREs for the relevant reservation classes,
making the assurance layer self-correcting under nonstationarity while preserving the same contract interface. This enables CSPs to continuously refine AIaaS assurance without violating contractual abstractions.

Finally, federation requires settlement aligned with tail guarantees. Byte-based or CPU-second-based revenue allocation is misaligned because
tail risk may be dominated by a single domain’s uncertainty. Define the bound-implied risk score at the contract deadline:
\begin{equation}
\mathcal{K}_u(\tau_u)\triangleq -\log \widehat{\Prob}\{W_u>\tau_u\},
\end{equation}
where \(\widehat{\Prob}\{\cdot\}\) is the TRE-composed value in \eqref{eq:tandem_bound}. A domain’s marginal tail-risk contribution is
computed by the sensitivity of \(\mathcal{K}_u\) to its contract parameters \((R_d,T_d,\kappa_d,\eta_d)\) at the deployed reservation.
This yields interpretable attribution without internal disclosure. Let \(P_u\) denote the tenant payment and \(\Pi_u\) the penalty upon an
EVT-audited breach. We allocate revenue shares proportionally to positive marginal contributions (domains that increase \(\mathcal{K}_u\)) and allocate penalties proportionally to negative marginal contributions (domains whose degradation reduces \(\mathcal{K}_u\)).%

\section{Simulation Results}\label{sec:sim}

We evaluate whether the proposed TRE-based management plane turns federated AIaaS into an enforceable managed service. The simulations target three questions: 
(i) can the system maintain tail-delay compliance under rising load, 
(ii) can it provide tenant isolation under bursty/adversarial traffic, and 
(iii) can it support interpretable accountability by attributing tail-risk increases to the domains responsible for degradation. 
\subsection{Setup}
We use packet-level Monte-Carlo simulation of FIFO queues. Across all experiments we report p99.9 delay ($p=0.999$) against a deadline $\tau=30$ (time units), and we compare best-effort operation (no contract-aware control) against TRE-managed operation (contract-aware admission/reservation). For reproducibility, we log the selected acceptance factor $\alpha$ at each offered load together with the resulting empirical $Q_{0.999}(D)$ values.

The end-to-end path is a tandem of $D=3$ single-server domains. Arrivals to the first domain are Poisson with rate $\lambda$. Domain $d$ provides exponential service with rate $\mu_d$ and includes a fixed processing/propagation shift $T_d$. End-to-end delay is the sum of per-domain sojourn times (waiting + service) plus the per-domain shifts. We use
\[
\boldsymbol{\mu}=[1.0,\,1.15,\,1.25], \qquad \mathbf{T}=[0.6,\,0.5,\,0.4],
\]
and we normalize offered load as $\rho=\lambda/\min_d \mu_d$.

For each Monte-Carlo trial we simulate $N_{\mathrm{pkt}}=6000$ packets and compute the empirical tail quantile $Q_{0.999}(D)$ of end-to-end delay using the sample quantile. Each operating point is averaged over $N_{\mathrm{MC}}$ independent trials (user-adjustable in the simulator). Random seeds are fixed to ensure repeatability.

We consider three scenarios. First, we sweep $\rho$ over a grid of 10 points in $[0.55,\,0.98]$. The best-effort baseline admits all traffic. TRE-managed operation applies admission control through an acceptance factor $\alpha\in(0,1]$ (admitted rate $\lambda_{\mathrm{adm}}=\alpha\lambda$). For each offered load, $\alpha$ is chosen by a bisection search (18 iterations) to maximize admission while satisfying the conservative tail target
\[
Q_{0.999}(D)\le 0.985\,\tau,
\]
evaluated on a single packet-level run; the resulting $\alpha$ is then re-evaluated with Monte-Carlo averaging.

Second, we evaluate isolation in a single shared bottleneck domain with service rate $\mu=\min_d \mu_d$. Two tenants share the domain: a victim stream with Poisson arrivals at rate $\lambda_v=0.55\,\mu$, and an attacker stream with fixed mean rate $\lambda_a=0.12\,\mu$ but increasing burstiness. Burstiness is controlled by a peak-to-mean ratio $b\in[1,8]$ using a correlated ON/OFF arrival process: during ON periods the attacker rate is $\lambda_{\mathrm{on}}=b\lambda_a$; during OFF periods it is $\lambda_{\mathrm{off}}=\lambda_a/b$; ON/OFF periods alternate with mean run lengths that increase with $b$ (longer bursts for larger $b$). We compare shared FIFO multiplexing (single queue, no isolation) against per-tenant reservation (separate queues with fixed service shares). In reservation mode, the victim receives a dedicated share $s_v=0.85$ (service rate $\mu_v=s_v\mu$) and the attacker receives $\mu_a=(1-s_v)\mu$.

Third, to simulate domain degradation, we scale service rates by a factor $s\in[0.60,1.0]$ (6 points), where $s=1$ is nominal and smaller values represent degraded service. We fix the offered load near the operating knee at $\lambda=0.85\,\min_d\mu_d$. We report the increase in mean p99.9 delay
\[
\Delta Q_{0.999}(D)=\mathbb{E}[Q_{0.999}(D)\!\mid \!s]-\mathbb{E}[Q_{0.999}(D)\!\mid \!s\!=\!1].
\]
To obtain marginal attribution, we compute (i) the total increase when all domains are degraded by $s$, and (ii) the increase when only domain $d$ is degraded by $s$; then the marginal contributions per-domain are normalized to sum up the total increase in each $s$. Table~\ref{tab:simparams} summarizes the parameters used in the plots.

\begin{table}[t]
\centering
\caption{Simulation parameters used in the reported experiments.}
\label{tab:simparams}
\begin{tabular}{@{}ll@{}}
\toprule
Tandem domains & $D=3$ \\
Service rates & $\boldsymbol{\mu}=[1.0,\,1.15,\,1.25]$ \\
Fixed shifts & $\mathbf{T}=[0.6,\,0.5,\,0.4]$ \\
Tail percentile & $p=0.999$ (p99.9) \\
Deadline & $\tau=30$ \\
Packets per trial & $N_{\mathrm{pkt}}=6000$ \\
Monte-Carlo trials & $N_{\mathrm{MC}}$ (10000) \\
Load grid & $\rho\in[0.55,0.98]$ (10 points) \\
TRE guard factor & $0.985$ \\
Isolation: victim load & $\lambda_v=0.55\,\mu$ \\
Isolation: attacker mean & $\lambda_a=0.12\,\mu$ \\
Isolation: victim share & $s_v=0.85$ \\
Burstiness grid & $b\in[1,8]$ (8 points) \\
Degradation grid & $s\in[0.60,1.0]$ (6 points) \\
\bottomrule
\end{tabular}
\vspace{-.4cm}
\end{table}

\subsection{Results}

Fig.~\ref{fig:load} reports the estimated p99.9 end-to-end delay $Q_{0.999}(D)$ versus normalized offered load $\rho$. Under best-effort operation, the tail quantile remains moderate at low-to-mid loads but rises sharply as $\rho$ approaches saturation, quickly exceeding the deadline $\tau$. This is the classical tail-amplification regime: once utilization becomes high, rare backlog events persist long enough to dominate extreme percentiles, so the service can look acceptable on average while violating p99.9 systematically. In contrast, TRE-managed operation avoids this cliff by adapting the admitted rate to preserve headroom. The curve staying close to $\tau$ indicates that the management plane is not merely improving mean latency; it is explicitly preventing tail runaway and thus maintaining a predictable worst-case user experience. Operationally, this supports the paper’s premise that federated AIaaS cannot be treated as best-effort transport: tail compliance requires an explicit control knob that trades throughput for reliability when the system is stressed.


Fig.~\ref{fig:isolation} evaluates robustness to adversarial burstiness. Under shared FIFO multiplexing, increasing attacker burstiness produces pronounced inflation of the victim’s p99.9 delay even though the attacker’s mean rate is held fixed. This is a key point for multi-tenant AIaaS: it is not the average load alone that determines tail compliance, but the burst structure of competing tenants. Correlated bursts create periods of sustained overload that force the victim to wait behind attacker backlog, producing a strong cross-tenant coupling at extreme percentiles. Under per-tenant reservation, the victim tail remains essentially flat over the full burstiness range, indicating that the victim’s tail behavior is governed primarily by its own traffic and reserved capacity rather than by attacker burst patterns. This directly supports the isolation claim: enforcing per-tenant shares prevents burst-induced tail inflation and converts a fragile shared-queue regime into a predictable service envelope, which is essential for offering contractual p99/p99.9 guarantees in a multi-tenant setting.

Fig.~\ref{fig:settlement} connects performance assurance to accountability. As domains are degraded (smaller service-rate scaling factor $s$), the increase in tail delay $\Delta Q_{0.999}(D)$ grows, confirming that extreme percentiles are highly sensitive to even moderate capacity loss. The marginal decomposition further shows that not all domains contribute equally to the tail increase: the most capacity-critical domain(s) account for a disproportionate share of $\Delta Q_{0.999}(D)$, consistent with the intuition that bottlenecks dominate end-to-end tail behavior. This provides two practical implications for federated operation. First, it enables audit-ready diagnostics: when tail violations occur, the attribution points to where degradation most strongly impacts the end-to-end objective. Second, it supports incentive-compatible federation: penalties or credits can be tied to each domain’s marginal tail-risk contribution rather than to coarse end-to-end metrics that unfairly blame all parties equally. In this sense, the figure demonstrates that the proposed assurance layer is not only predictive (bounding risk) but also operational (supporting accountability under multi-operator composition).


\begin{figure}[t]
  \centering
  \includegraphics[width=\linewidth]{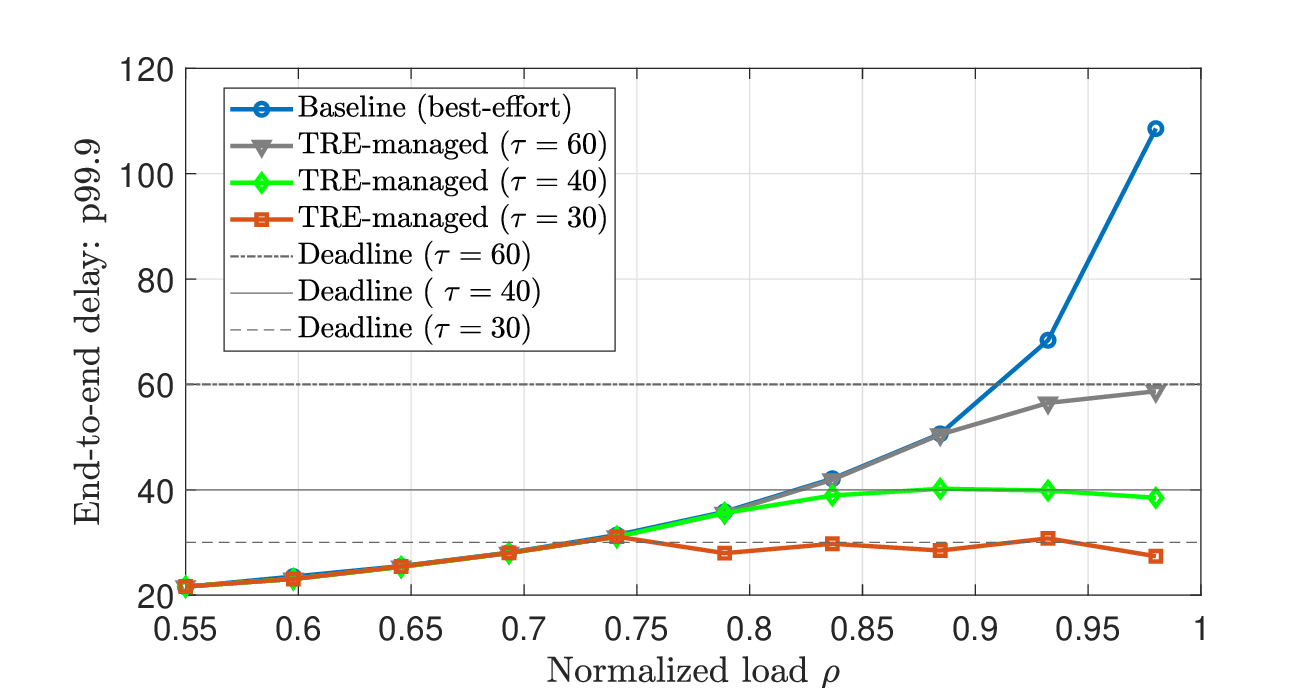}
    \vspace{-0.6cm}
  \caption{Estimated p99.9 end-to-end delay versus normalized offered load $\rho$.}
  \label{fig:load}
     \vspace{-0.4cm}
\end{figure}


\begin{figure}[t]
  \centering
  \includegraphics[width=\linewidth]{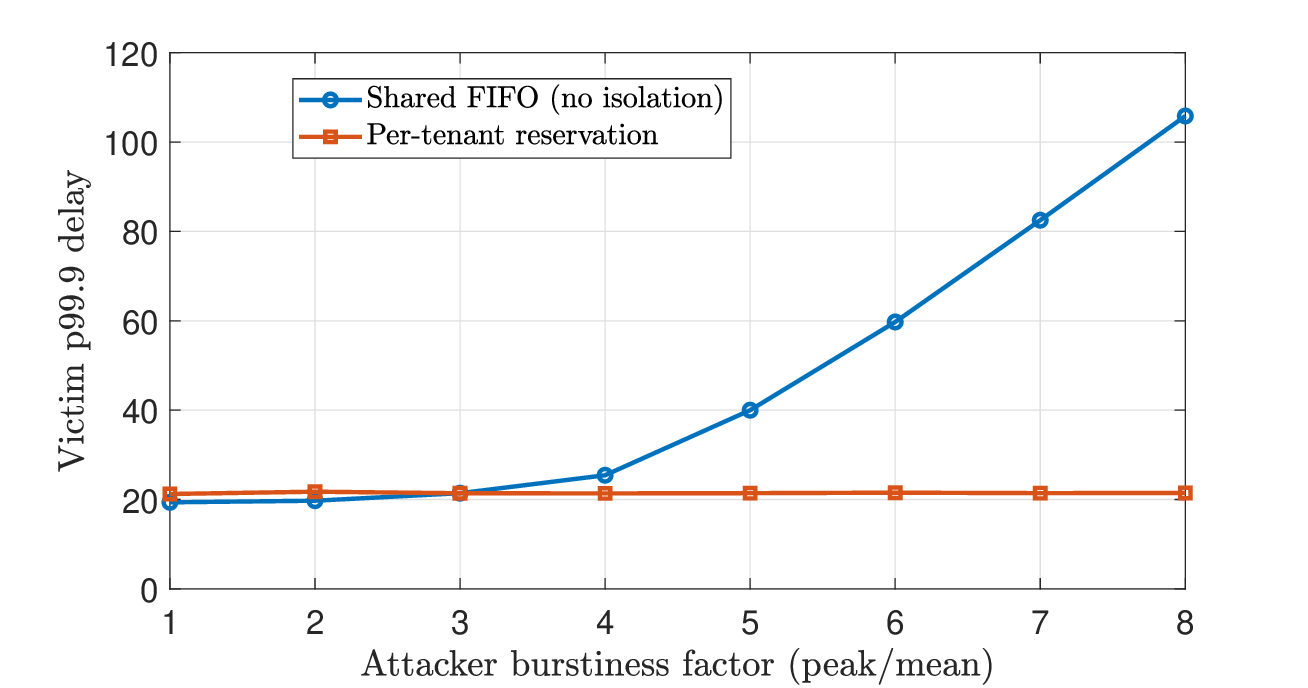}
  \vspace{-0.6cm}
  \caption{Isolation under burstiness.}
  \label{fig:isolation}
  \vspace{-0.4cm}
\end{figure}

\begin{figure}[t]
  \centering
  \includegraphics[width=\linewidth]{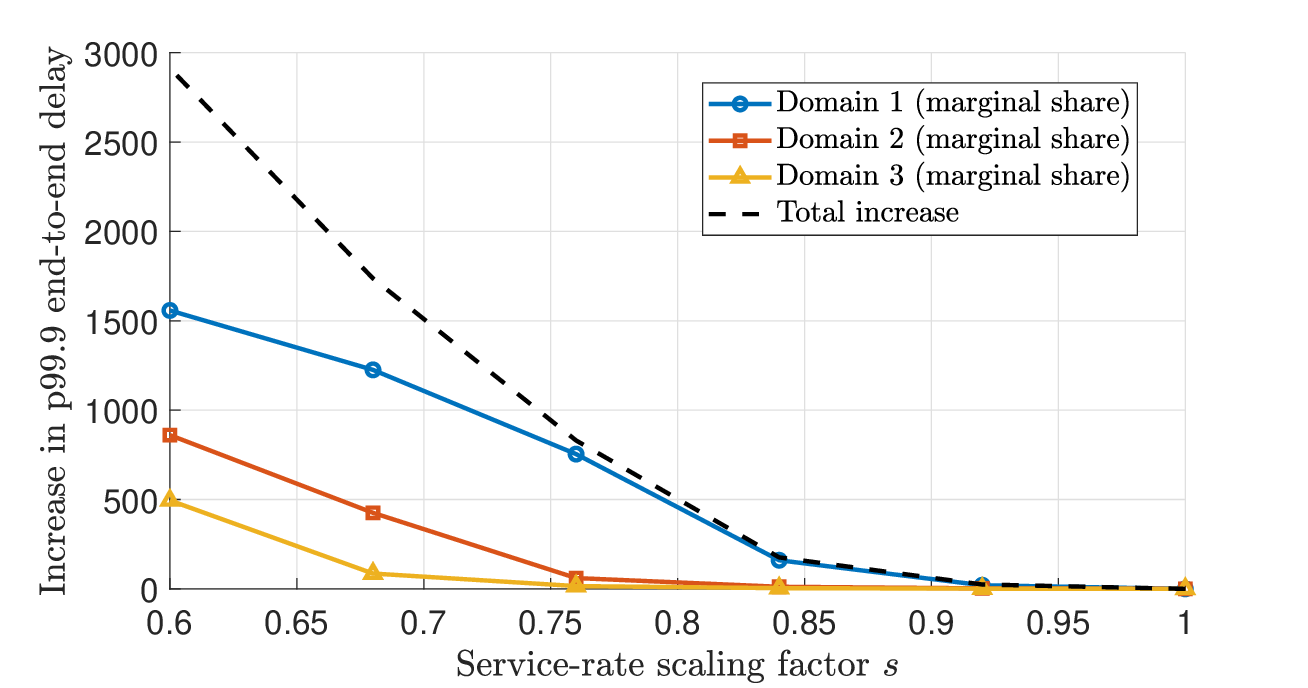}
    \vspace{-0.6cm}
  \caption{Marginal tail-risk attribution under controlled degradation.}
  \label{fig:settlement}
       \vspace{-0.5cm}
\end{figure}


\section{Conclusion}
This paper proposes an assurance-oriented management and orchestration framework that makes \emph{federated} AIaaS enforceable across multiple administrative domains. We introduced TREs as signed, composable per-domain contracts that support end-to-end p99/p99.9 reasoning while requiring only disclosure-minimal parameters. Building on these contracts, we develop a federated orchestration method that decomposes an end-to-end tail-risk budget into per-domain obligations, enabling confidentiality-preserving coordination and strict multi-tenant isolation. A telemetry-driven auditing layer continuously recalibrates tail models and supports accountability by attributing end-to-end tail-risk changes to individual domains. Simulations confirm improved p99.9 reliability under overload, robust isolation under bursty interference, and stable attribution under controlled multi-domain degradation.

\bibliographystyle{IEEEtran}
\bibliography{IEEEabrv,refs}

\end{document}